\RequirePackage{amsmath}
\documentclass[runningheads]{llncs}
\usepackage{amsfonts,amssymb}
\usepackage{geometry}
\usepackage{mathrsfs}
\usepackage{enumerate}
\usepackage{mathtools}
\usepackage{gastex}
\usepackage{tabularx}
\usepackage{xcolor}
\usepackage{pdfsync}

\newcommand{\overlap}{\odot}

\newcommand{\eps}{\varepsilon}

\newcommand{\cB}{{\mathcal B}}

\newcommand{\cD}{{\mathcal D}}

\newcommand{\cN}{{\mathcal N}}

\newcommand{\cW}{{\mathcal W}}

\newcommand{\Sig}{\Sigma}
\newcommand{\bi}{\begin{itemize}}
\newcommand{\ei}{\end{itemize}}
\newcommand{\be}{\begin{enumerate}}
\newcommand{\ee}{\end{enumerate}}
\newcommand{\bd}{\begin{description}}
\newcommand{\ed}{\end{description}}
\newcommand{\bq}{\begin{quote}}
\newcommand{\eq}{\end{quote}}

\title{State Complexity of Overlap Assembly\thanks{This work was supported by the Natural Sciences and Engineering Research Council of Canada under grants No.~OGP0000871 and R2824A01, and by the National Science Centre, Poland, under project number 2014/15/B/ST6/00615.}}

\author{Janusz~Brzozowski \inst{1}, 
Lila Kari\inst{1}, 
Bai Li \inst{1},
\and Marek Szyku{\l}a \inst{2}
}

\authorrunning{Brzozowski, Kari, Li, Szyku{\l}a}

\institute{David R. Cheriton School of Computer Science, University of Waterloo, \\
Waterloo, ON, Canada N2L 3G1\\
\email{brzozo@uwaterloo.ca, lila@uwaterloo.ca, bai.li.2005@gmail.com}
\and
Institute of Computer Science, University of Wroc{\l}aw,\\
Joliot-Curie 15, PL-50-383 Wroc{\l}aw, Poland\\
{\tt msz@cs.uni.wroc.pl}
}

\begin{document}
\maketitle

\begin{abstract}
The \emph{state complexity} of a regular language $L_m$ is the number $m$ of states in a minimal deterministic finite automaton (DFA) accepting $L_m$. The state complexity of a regularity-preserving binary operation on regular languages is defined as the maximal state complexity of the result of the operation where the two operands range over all languages of state complexities $\le m$ and $\le n$, respectively. We find a tight upper bound on the state complexity of the binary operation \emph{overlap assembly} on regular languages. This operation was introduced by Csuhaj-Varj\'{u}, Petre, and Vaszil to model the process of self-assembly of two linear DNA strands into a longer DNA strand, provided that their ends ``overlap''.
We prove that the state complexity of the overlap assembly of languages $L_m$ and $L_n$, where $m\ge 2$ and $n\ge1$, is at most $2 (m-1) 3^{n-1} + 2^n$. Moreover, for $m \ge 2$ and $n \ge 3$ there exist languages $L_m$ and $L_n$ over an alphabet of size $n$ whose overlap assembly meets the upper bound and this bound cannot be met with smaller alphabets. Finally, we prove that $m+n$ is a tight upper bound on the overlap assembly of unary languages, and that there are binary languages whose overlap assembly has exponential state complexity at least $m(2^{n-1}-2)+2$.
\medskip

\end{abstract}

\section{Introduction}\label{sec:intro}

The \emph{state complexity} of a~regular language is the number of states in a~minimal deterministic finite automaton (DFA) accepting the language. 
The state complexity of a~regularity-preserving binary operation on regular languages is defined as the maximal state complexity of the result of the operation when the operands range over all languages of state complexities $\le m$ and $\le n$; it is a~function of $m$ and $n$.
State complexity was introduced by Maslov~\cite{Mas70} in 1970, but his short paper was relatively unknown for many years. 
Maslov stated without proof that the state complexity of the (Kleene) star of a~language $L_n$ of state complexity $n$ is $2^{n-1}+2^{n-2}$, that of reversal is $2^n$, that of concatenation of languages $L_m$ and $L_n$ of state complexities $m$ and $n$, respectively, is $(m-1)2^n+ 2^{n-1}$, and that of union is $mn$.
A~more complete study of state complexity including proofs was presented by Yu, Zhuang, and Salomaa~\cite{YZS94} in 1994. They proved that the state complexity of intersection is also $mn$.
The same bound also holds for other binary Boolean functions such as symmetric difference and difference~\cite{Brz13}.
Since the publication of the paper by Yu, Zhuang, and Salomaa, many authors have written on this subject; for an~extensive bibliography see the recent surveys~\cite{Brz18,GMRY16}.
In particular, the state complexities of the so-called basic operations, namely Boolean operations, concatenation, star and reversal in various subclasses of the class of regular languages have been studied~\cite{Brz18}.

In this paper, we consider the state complexity of a~biologically inspired binary word and language operation called \emph{overlap assembly}. Formally, overlap assembly is a~binary operation which, when applied to two input words $xy$ and $yz$ (where $y$ is their nonempty {\it overlap}), produces the output $xyz$. As a~formal language operation, overlap assembly was introduced by Csuhaj-Varj\'u, Petre, and Vaszil~\cite{CPV07} under the name ``self-assembly'', and studied by Enaganti, Ibarra, Kari and Kopecki~\cite{EIKK16,EIKK17}. A~particular case of overlap assembly, called {\it chop operation}, where the overlap consists of a~single letter, was studied in~\cite{HolzerJacobi2011,HolzerJacobi2012}, and generalized to an~arbitrary length overlap in~\cite{HoJaKu17}. Other similar operations have been studied in the literature, such as the {\it short concatenation} \cite{CarausuPaun81}, which uses only the maximum-length (possibly empty) overlap $y$ between operands, the Latin product of words \cite{Golan92} where the overlap consists of only one letter, and the operation $\bigotimes$ which imposes the restriction that the non-overlapping part $xz$ is not empty \cite{ItoLischke07}. Overlap assembly can also be considered as a~particular case of semantic shuffle on trajectories with trajectory $0^*\sigma^+1^*$ \cite{Domaratzki2009}\footnote{Informally, during a~shuffle between two words with a~trajectory over $\{0, 1, \sigma\}^+$, the symbols of the trajectory are interpreted as follows:
 0 (respectively 1) signifies that the corresponding letter from the first (respectively second) word is retained, and $\sigma$ signifies that a~letter from the first word is retained, provided it coincides with the corresponding letter in the second word.}, or as a~generalization of the operation $\bigodot_N$ from \cite{Domaratzki2009} which imposes the length of the overlap to be at least $N$.

The study of overlap assembly as a~formal language operation was initiated in the context of research on DNA-based information and DNA-based computation, as a~formalization of a~biological lab procedure that combines short linear DNA strands into longer ones, provided that their ends ``overlap''. The process of overlap assembly is enabled by an~active agent called the DNA polymerase enzyme, which has the property of being able to extend DNA strands, under certain conditions. Other DNA bio-operations enabled by the action of the DNA polymerase enzyme, which have been modeled and studied as formal language operations, include hairpin completion and its inverse operation, hairpin reduction~\cite{CMM06,Kop11,MMM09,MaMi07}, overlapping concatenation~\cite{MPPR03}, and directed extension~\cite{EKK15}.
 Experimentally, (parallel) overlap assembly of DNA strands under the action of the DNA polymerase enzyme was used for gene shuffling in, e.g.,~\cite{Ste94}. In the context of experimental DNA computing, overlap assembly was used in, e.g., \cite{CFLL99,FCLL00,KOTL97,OLL97} for the formation of combinatorial DNA or RNA libraries. Overlap assembly can also be viewed as modeling a~special case of an~experimental lab procedure called cross-pairing PCR, introduced in \cite{FrMa11} and studied in, e.g., \cite{Fra05,FGLM05,FMGL06,MaFr08}.

In this paper, we investigate the state complexity of overlap assembly as a~binary operation on regular languages. The paper is organized as follows. Section~\ref{sec:overlap} describes the biological motivation of overlap assembly. Section~\ref{sec:NFA_construction} introduces our notation and describes the construction of an~NFA that accepts the results of overlap assembly of two regular languages, given by their accepting DFAs.
 In Section~\ref{sec:general} we prove that the state complexity of the overlap assembly of languages $L_m$ and $L_n$, where $m\ge 2$ and $n\ge1$, is at most $2 (m-1) 3^{n-1} + 2^n$ (Theorem~\ref{thm:upper_bound}). Moreover, for $m \ge 2$ and $n \ge 3$ there exist languages $L_m$ and $L_n$ over an~alphabet of size $n$ whose overlap assembly meets the upper bound (Theorem~\ref{thm:meets_upper_bound}) and, in addition, this bound cannot be met with smaller alphabets (Theorem~\ref{thm:min_alphabet}).
Section~\ref{sec:unary} proves that $m+n$ is a~tight upper bound on the descriptional complexity of the overlap assembly of two unary regular languages $L_m$ and $L_n$ (Theorem \ref{thm:unary}), and in Section~\ref{sec:binary} we show that in the case of a~binary alphabet the state complexity can be at least $m(2^{n-1}-2)+2$, thus is already exponential in $n$.

A shorter version of this work not containing the results about unary and binary alphabets has appeared in~\cite{BKLS18}.

\section{Overlap Assembly}\label{sec:overlap}

The bio-operation of overlap assembly was intended to model the procedure whereby short DNA single strands can be concatenated (assembled) together into longer strands under the action of the enzyme DNA polymerase, provided they have ends that ``overlap''. Recall that DNA single strands are oriented words from the DNA alphabet $\Delta = \{A, C, G, T\}$, where one end of a~strand is labeled by $5'$ and the other by $3'$. Watson/Crick (W/C) complementarity of DNA strands couples $A$ to $T$ and $C$ to $G$ and acts as follows: Given two W/C single strands, of opposite orientation, and whose letters are complementary at each position, the W/C complementarity of DNA strands binds the two single strands together by covalent bonds, to form a~DNA double strand. The W/C complementarity of DNA strands has been traditionally modeled \cite{Hussini2002,Kari2002} as an~antimorphic involution $\theta\colon \Delta^* \longrightarrow \Delta^*$, that is, an~involution on $\Delta$ ($\theta^2$ is the identity on $\Delta$) extended to an~antimorphism on $\Delta^*$, whereby $\theta(uv) = \theta(v) \theta(u)$ for all $u, v \in \Delta^*$. In this formalism, the W/C complement of a~DNA strand $u \in \Delta^+$ is $\theta(u)$.

Using the convention that a~word $x$ over the DNA alphabet represents the DNA single strand $x$ in the $5'$ to $3'$ direction (usually depicted as the top strand of a~double DNA strand), the {\it overlap assembly} of a~strand $uv$ with a~strand $\theta(w)\theta(v)$ first forms a~partially double-stranded DNA molecule, where the substrand $v$ in $uv$ binds to the substrand $\theta(v)$ in $\theta(w)\theta(v)$; see Figure~\ref{fig:OC_graphical}(a). The DNA polymerase enzyme will then extend the $3'$ end of $uv$ with the strand $w$; see Figure~\ref{fig:OC_graphical}(b). Similarly, the $3'$ end of $\theta(w)\theta(v)$ will be extended, resulting in a~full double strand whose upper strand is $5' - uvw - 3'$, and bottom strand is $5' -\theta(w)\theta(v)\theta(u)-3'$, see Figure~\ref{fig:OC_graphical}(c). Thus, in principle, the overlap assembly between $u v$ and $ \theta(w)\theta(v)$ results in the strands $uvw$ and $\theta(uvw) = \theta(w)\theta(v)\theta(u)$. 

\begin{figure}[htb]\centering
{\unitlength 8.0pt\large\begin{picture}(37,24)(-1,-3)
\put(5,20){
\node[Nframe=n]()(-5,0){(a)}
\node[Nframe=n]()(5,1){$u$}
\node[Nframe=n]()(15,1){$v$}
\node[Nframe=n]()(0,1.5){\normalsize$5'$}
\node[Nframe=n]()(20,1.5){\normalsize$3'$}
\drawline[AHnb=0](0,-.5)(0,.5)
\drawline[AHnb=0,linewidth=.4](0,0)(10,0)
\drawline[AHnb=0](10,-.5)(10,.5)
\drawline[AHnb=0,linewidth=.4](10,0)(20,0)
\drawline[AHnb=0](20,.5)(20,-.5)
\drawline[AHnb=0,dash={0.5 0.5}0](20,0)(30,0)
}\put(5,18){
\node[Nframe=n]()(10,-1.5){\normalsize$3'$}
\node[Nframe=n]()(30,-1.5){\normalsize$5'$}
\node[Nframe=n]()(15,-1.5){$\theta(v)$}
\node[Nframe=n]()(25,-1.5){$\theta(w)$}
\drawline[AHnb=0](10,.5)(10,-.5)
\drawline[AHnb=0,linewidth=.4](10,0)(20,0)
\drawline[AHnb=0](20,.5)(20,-.5)
\drawline[AHnb=0,linewidth=.4](20,0)(30,0)
\drawline[AHnb=0](30,.5)(30,-.5)
}

\put(5,11){
\node[Nframe=n]()(0,1.5){\normalsize$5'$}
\node[Nframe=n]()(30,1.5){\normalsize$3'$}
\node[Nframe=n]()(-5,0){(b)}
\node[Nframe=n]()(5,1){$u$}
\node[Nframe=n]()(15,1){$v$}
\node[Nframe=n]()(25,1){$w$}
\drawline[AHnb=0](0,-.5)(0,.5)
\drawline[AHnb=0,linewidth=.4](0,0)(10,0)
\drawline[AHnb=0](10,-.5)(10,.5)
\drawline[AHnb=0,linewidth=.4](10,0)(20,0)
\drawline[AHnb=0](20,.5)(20,-.5)
\drawline[AHnb=0,linewidth=.4](20,0)(30,0)
\drawline[AHnb=0](30,.5)(30,-.5)
}\put(5,9){
\node[Nframe=n]()(10,-1.5){\normalsize$3'$}
\node[Nframe=n]()(30,-1.5){\normalsize$5'$}
\node[Nframe=n]()(15,-1.5){$\theta(v)$}
\node[Nframe=n]()(25,-1.5){$\theta(w)$}
\drawline[AHnb=0](10,.5)(10,-.5)
\drawline[AHnb=0,linewidth=.4](10,0)(20,0)
\drawline[AHnb=0](20,.5)(20,-.5)
\drawline[AHnb=0,linewidth=.4](20,0)(30,0)
\drawline[AHnb=0](30,.5)(30,-.5)
}

\put(5,2){
\node[Nframe=n]()(0,1.5){\normalsize$5'$}
\node[Nframe=n]()(30,1.5){\normalsize$3'$}
\node[Nframe=n]()(-5,0){(c)}
\node[Nframe=n]()(5,1){$u$}
\node[Nframe=n]()(15,1){$v$}
\node[Nframe=n]()(25,1){$w$}
\drawline[AHnb=0](0,-.5)(0,.5)
\drawline[AHnb=0,linewidth=.4](0,0)(10,0)
\drawline[AHnb=0](10,-.5)(10,.5)
\drawline[AHnb=0,linewidth=.4](10,0)(20,0)
\drawline[AHnb=0](20,.5)(20,-.5)
\drawline[AHnb=0,linewidth=.4](20,0)(30,0)
\drawline[AHnb=0](30,.5)(30,-.5)
}\put(5,0){
\node[Nframe=n]()(0,-1.5){\normalsize$3'$}
\node[Nframe=n]()(30,-1.5){\normalsize$5'$}
\node[Nframe=n]()(5,-1.5){$\theta(u)$}
\node[Nframe=n]()(15,-1.5){$\theta(v)$}
\node[Nframe=n]()(25,-1.5){$\theta(w)$}
\drawline[AHnb=0](0,.5)(0,-.5)
\drawline[AHnb=0,linewidth=.4](0,0)(10,0)
\drawline[AHnb=0](10,.5)(10,-.5)
\drawline[AHnb=0,linewidth=.4](10,0)(20,0)
\drawline[AHnb=0](20,.5)(20,-.5)
\drawline[AHnb=0,linewidth=.4](20,0)(30,0)
\drawline[AHnb=0](30,.5)(30,-.5)
} 
\end{picture}}
\caption{(a) The two input DNA single-strands, $uv$ and $\theta(w)\theta(v)$ bind to each other through their complementary segments $v$ and $\theta(v)$, forming a~partially double-stranded DNA complex. (b) DNA polymerase extends the $3'$ end of the strand $uv$. (c) DNA polymerase extends the $3'$ end of the other strand. The resulting DNA double strand is considered to be the output of the {\it overlap assembly} of the two input single strands.}
\label{fig:OC_graphical}
\end{figure}
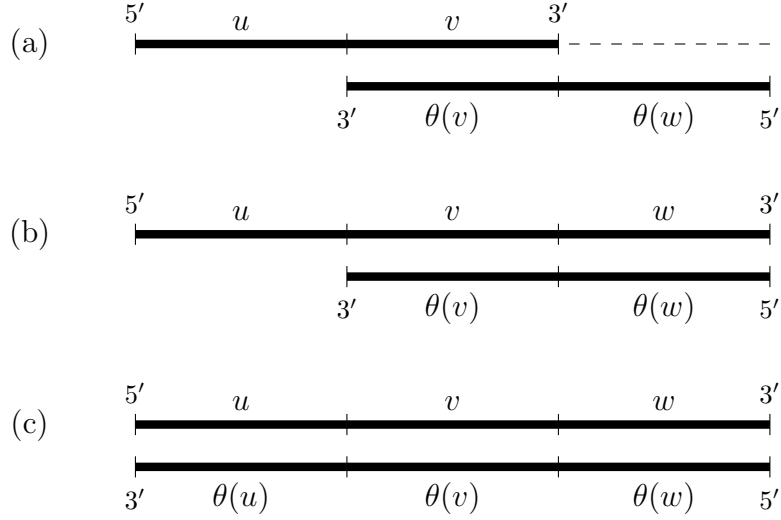

Assuming that all involved DNA strands are initially double-stranded, that is, whenever the strand $x$ is available, its W/C complement $\theta(x)$ is also available, this model was further simplified \cite{CPV07} as follows: Given words $x, y$ over an~alphabet $\Sigma$, the {\it overlap assembly of $x$ with $y$} is defined as:
\[
x \overlap y = \{z \in \Sigma^+ \mid \exists u,w \in \Sigma^*, \exists v \in\Sigma^+: x = uv, y = vw , z = uvw\}.
\]

This can be naturally generalized to languages: Given languages $L_m$ and $L_n$ of state complexities $m$ and $n$, respectively, the overlap assembly of $L_m$ and $L_n$ is defined as:
$$ L_m \overlap L_n = \{ z \mid z = x \overlap y,\,  x \in L_m, y \in L_n \}. $$

\section{An~$\eps$-NFA for Overlap Assembly}
\label{sec:NFA_construction}

A~\emph{deterministic finite automaton (DFA)} is a~quintuple
$\cD=(Q, \Sigma, \delta, q_0,F)$, where
$Q$ is a~finite non-empty set of \emph{states},
$\Sig$ is a~finite non-empty \emph{alphabet},
$\delta\colon Q\times \Sig\to Q$ is the \emph{transition function},
$q_0\in Q$ is the \emph{initial} state, and
$F\subseteq Q$ is the set of \emph{final} states.
We extend $\delta$ to functions $\delta\colon Q\times \Sig^*\to Q$ and $\delta\colon 2^Q\times \Sig^*\to 2^Q$ as usual.
A~DFA $\cD$ \emph{accepts} a~word $w \in \Sigma^*$ if ${\delta}(q_0,w)\in F$. The language accepted by $\cD$ is denoted by $L(\cD)$. If $q$ is a~state of $\cD$, then the language $L_q(\cD)$ of $q$ is the language accepted by the DFA $(Q,\Sigma,\delta,q,F)$. 
A~state is \emph{empty} (or \emph{dead} or a~\emph{sink state}) if its language is empty. Two states $p$ and $q$ of $\cD$ are \emph{equivalent} if $L_p(\cD) = L_q(\cD)$. 
A~state $q$ is \emph{reachable} if there exists $w\in\Sig^*$ such that $\delta(q_0,w)=q$.
A~DFA $\cD$ is \emph{minimal} if it has the smallest number of states and the smallest alphabet among all DFAs accepting $L(\cD)$.
It is well known that a~DFA is minimal if it uses the smallest alphabet, all of its states are reachable, and no two states are equivalent.

A~\emph{nondeterministic finite automaton (NFA)} is a~quintuple
$\mathcal{N}=(R,\Sigma,\eta,I,F)$, where $R$, $\Sigma$, and $F$ are as $Q$, $\Sigma$, and $F$ in a~DFA respectively,
$\eta\colon R\times \Sigma\to 2^{R}$, and
$I\subseteq R$ is the \emph{set of initial states}. 
Each triple $(p,a,q)$ with $p,q\in R$, $a\in\Sig$ is a~\emph{transition} if $q\in \eta(p,a)$.
A~sequence $((p_0,a_0,q_0), (p_1,a_1,q_1), \dots, (p_{k-1},a_{k-1},q_{k-1}))$
of transitions, where $p_{i+1}=q_i$ for $i=0, \dots, k-2$ is a~\emph{path} in $\cN$.
The word $a_0a_1\cdots a_{k-1}$ is the word \emph{spelled} by the path. 
A~word $w$ is \emph{accepted} by $\cN$ if there exists a~path with $p_0\in I$ and $q_{k-1}\in F$ that spells $w$.
If $q\in \eta(p,a)$ we also use the notation $p \xrightarrow{a} q$. We extend this notation also to words, and write 
$p \xrightarrow{w} q$ for $w\in\Sig^*$.
An~\emph{$\varepsilon$-NFA} is an~NFA in which transitions under the empty word $\varepsilon$ are also permitted.

\medskip

Given any two DFAs, we construct an~$\eps$-NFA that recognizes the overlap assembly of the languages accepted by the DFAs. This proves constructively that the family of regular languages is closed under overlap assembly.

Let $\cD_m = (Q_m, \Sigma, \delta_m, 0, F)$ and $\cD'_n = (Q'_n, \Sigma, \delta'_n, 0', F')$ be two DFAs with $\cD_m$ recognizing $L_m$ and $\cD'_n$ recognizing $L'_n$, where $F=\{f_1,\dots,f_h\}$ and 
$F'=\{f'_1,\dots,f'_{h'}\}$. Let $Q_m = \{0, \ldots, m-1\}$, $Q'_n = \{0',\ldots,(n-1)'\}$, and let $0$ and $0'$ be the initial states. 
We claim that the NFA $\cN$, constructed as shown below, accepts the result of the overlap assembly of $L_m$ and $L_n'$.

The NFA is defined as $\cN = (R, \Sigma, \eta,\{r_0\}, F_\cN)$ where the set of states is $R = (Q_m \cup \{t\}) \times (Q_n' \cup \{ s' \})$ with $s', t$ new symbols not occurring in $Q_m\cup Q_n'$, the initial state is $r_0 = (0,s')$, and the set of final states is $F_\cN = \{(t, q') \mid q' \in F'\}$.
Intuitively, the NFA simulates reading the word first by $\cD_m$, then by both $\cD_m$ and $\cD'_n$, and then by $\cD'_n$.
Hence the states in $R$ contain a~state of $\cD_m$ and a~state of $\cD'_n$.
The states with $s'$ indicate that the second DFA has not yet read any letter, while the states with $t$ indicate that the first DFA has finished its reading.
The set of transitions $\eta$ is defined below.
The informal explanations at the right of transition definitions assume two operands $uv \in L_m$ and $vw \in L'_n$ respectively. The word $z = uvw$ belongs to their overlap assembly.
\begin{enumerate}[i]
\item $\{(q_i, s') \xrightarrow{a} (q_j, s') \mid q_i \xrightarrow{a} q_j \in \delta_m\};$
read $u$.
\item $\{(q_i, s') \xrightarrow{a} (q_j, q_k') \mid q_i \xrightarrow{a} q_j \in \delta_m, \; 0' \xrightarrow{a} q_k' \in \delta'_n\};$
read the first letter of $v$.
\item $\{(q_i, q'_k) \xrightarrow{a} (q_j, q'_\ell) \mid q_i \xrightarrow{a} q_j \in \delta_m, \; q'_k \xrightarrow{a} q'_\ell \in \delta'_n\};$
read the remainder of $v$.
\item $\{(f_i, q'_k) \xrightarrow{\eps} (t, q_k') \mid f_i\in F, \; q'_k \in Q'_n \};$
$v$ has been read.
\item $\{(t, q'_k) \xrightarrow{a} (t, q'_\ell) \mid q'_k \xrightarrow{a} q'_\ell\in \delta'_n \};$
these rules read $w$.
\end{enumerate}

Figure~\ref{fig:NFA-example} illustrates the construction of such an~NFA, denoted by $\cN'$, for two particular two-state DFAs $\cD_2$ and $\cD'_2$    accepting the languages  $L(D_2)$ (all words over $\{a, b\}^*$ that have an~odd number of $a$s) and $L(D_2')$  (all words over $\{a, b\}^*$ that end in the letter $a$). Note that the overlap assembly of $L(D_2)$ and $L(D_2')$ is $L(D_2')$.

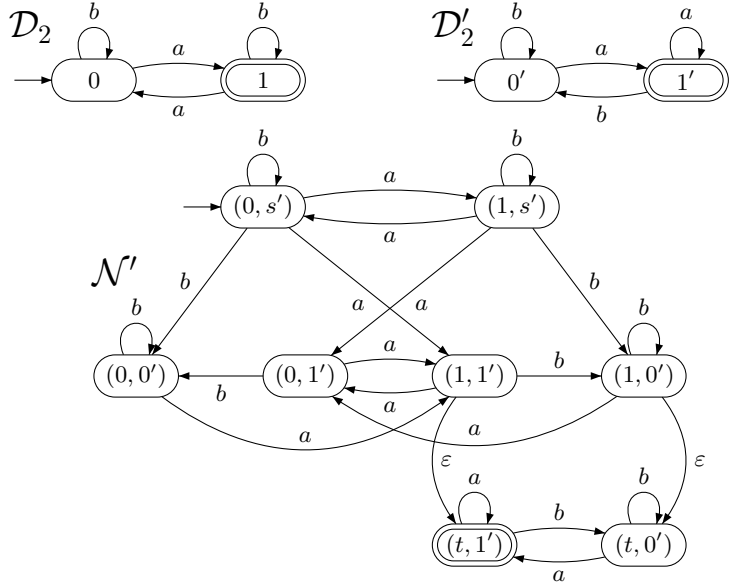
\begin{figure}[h]
\unitlength 8.0pt
\begin{center}\begin{picture}(38,30)(-8,-8)
\gasset{Nh=2.0,Nw=4,Nmr=1.25,ELdist=0.4,loopdiam=1.5}

\node(0)(-4,14){$0$}\imark(0)
\node(1)(4,14){$1$}\rmark(1)
\drawedge[curvedepth= .8,ELdist=.4](0,1){$a$}	
\drawedge[curvedepth= .8,ELdist=.4](1,0){$a$}
\drawloop(0){$b$}
\drawloop(1){$b$}
\node[Nframe=n](Dm)(-7,16.5){{\Large$\cD_2$}}
\node[Nframe=n](Dn)(13,16.5){{\Large$\cD'_2$}}

\node(0')(16,14){$0'$}\imark(0')
\node(1')(24,14){$1'$}\rmark(1')
\drawedge[curvedepth= .8,ELdist=.4](0',1'){$a$}	
\drawedge[curvedepth= .8,ELdist=.4](1',0'){$b$}
\drawloop(0'){$b$}
\drawloop(1'){$a$}

\node(0s')(4,8){$(0,s')$}\imark(0s')
\node(1s')(16,8){$(1,s')$}
\drawedge[curvedepth= .8,ELdist=.4](0s',1s'){$a$}	
\drawedge[curvedepth= .8,ELdist=.4](1s',0s'){$a$}
\drawloop(0s'){$b$}
\drawloop(1s'){$b$}

\node[Nframe=n](Nn)(-3,5){{\Large$\cN'$}}

\node(00')(-2,0){$(0,0')$}
\node(01')(6,0){$(0,1')$}
\node(11')(14,0){$(1,1')$}
\node(10')(22,0){$(1,0')$}
\drawedge[curvedepth= -3.5,ELdist=.4](00',11'){$a$}
\drawloop(00'){$b$}	
\drawedge[curvedepth= .8,ELdist=.4](01',11'){$a$}
\drawedge(01',00'){$b$}
\drawedge(11',10'){$b$}
\drawedge[curvedepth= 3.3,ELdist=-1](10',01'){$a$}
\drawedge[curvedepth= .8,ELdist=.4](11',01'){$a$}
\drawloop(10'){$b$}	

\drawedge[ELdist=-1.3](0s',00'){$b$}
\drawedge[ELdist=-1.2](0s',11'){$a$}
\drawedge(1s',01'){$a$}
\drawedge(1s',10'){$b$}

\node(t1')(14,-8){$(t,1')$}\rmark(t1')
\node(t0')(22,-8){$(t,0')$}
\drawedge[curvedepth= .8,ELdist=.4](t1',t0'){$b$}	
\drawedge[curvedepth= .8,ELdist=.4](t0',t1'){$a$}
\drawloop(t0'){$b$}
\drawloop(t1'){$a$}
\drawedge[curvedepth= -2,ELdist=.4](11',t1'){$\eps$}
\drawedge[curvedepth= 2,ELdist=.4](10',t0'){$\eps$}

\end{picture}\end{center}
\caption{An~example of an~NFA $\cN'$ that accepts the overlap assembly of the languages accepted by the DFAs $\cD_2$ (which accepts all words over $\{a, b\}^*$ that have an~odd number of $a$s) and $\cD'_2$ (which accepts all words over $\{a, b\}^*$  that end in the letter $a$).}
\label{fig:NFA-example}
\end{figure}

 In the automaton $\cN'$ of Figure~\ref{fig:NFA-example}, states $(0,s')$ and $(1,s')$ in the first row of the figure behave as specified in Rule~(i), using the transitions of $\cD_2$.
Rule~(ii) moves the states from the first row to the second row of the figure. 
In the second row, the transitions are those of the direct product of $\cD_2$ and $\cD'_2$, as directed by Rule~(iii).
Note that neither Rule~(i) nor Rule~(ii) can be used again since $s'$ does not appear as a~component of any state after Rule~(iii) is used.
When $\cN'$ is in a~state where the first component is $1$, which is a~final state of $\cD_2$, $\cN'$ can move to the next row following Rule~(iv) and change the first component of the state to $t$.
Note that Rule~(iii) cannot be used again since $t$ appears as the first component of every state after Rule~(iv) is used.
Finally, $\cN'$ moves to the third row and follows the transitions of $\cD'_2$.
Note that Rule~(iv) cannot be used again because of $t$. While the NFA $\cN'$ has eight states, converting it to a~DFA and minimizing this DFA results   in $D_2'$. 
The NFA $\cN'$ accepts the overlap assembly of $L(D_2)$ and $L(D_2')$. In general, the following result holds:

\begin{proposition}
Let $L_m$ and $L_n'$ be two regular languages accepted by the DFAs defined above, and let the NFA $\cN$ be the automaton constructed as above. 
NFA $\cN$ has the following properties:
\begin{enumerate}
\item
If $uv \in L_m$ and $vw \in L'_n$, then $r_0 \xrightarrow{uvw} r_f$ in $\cN$ where $r_f \in F_\cN$.
\item
If  $r_0 \xrightarrow{z} r_f$ in $\cN$, then there exist $u,w \in\Sig^*$, $v\in \Sig^+$ such that $z = uvw$, where $uv \in L_m$ and $vw \in L'_n$.
\item
$\cN$ accepts $L_m \overlap L'_n$.
\end{enumerate}
\end{proposition}
\begin{proof}
\begin{enumerate}
\item
For the first claim,  let $v = ax$, where $a\in\Sig$. If $uv \in L_m$  then  $0 \xrightarrow{uax} f_i$, for some $f_i\in F$ in $\cD_m$.
So there exist $q_i$ and $q_j$ in $Q_m$ such that 
$0 \xrightarrow{u} q_i \xrightarrow{a} q_j \xrightarrow{x} f_i$ in $\cD_m$.
Similarly, if $vw \in L_n$, then there exist $q'_k$ and $q'_\ell$ in $Q'_n$ such that  
$0' \xrightarrow{a} q'_k \xrightarrow{x} q'_\ell  \xrightarrow{w}  f'_j$, for some $f'_j \in F'$ in $\cD'_n$.

By construction we have in $\cN$: 
$$(0,s') \xrightarrow[(i)]{u} (q_i,s') \xrightarrow[(ii)]{a} (q_j,q'_k) \xrightarrow[(iii)]{x} (f_i,q'_\ell) 
\xrightarrow[(iv)]{\eps} (t,q'_\ell) \xrightarrow[(v)]{w} (t,f'_j),
$$
which proves our first claim.
\item
 Suppose that $r_0  \xrightarrow{z} r_f$ in $\cN$, where $r_f \in F_\cN$. By the construction of $\cN$, such a~path must proceed by $i$ applications of rule (i),  one application of rule (ii),  $j$ applications of rule (iii),  one $\eps$-transition via rule (iv), and $k$ applications of rule (v), where $i, j, k \geq 0$.
Thus there exist $u$, $v$, and $w$ in $\Sig^*$ such that  $z = uvw$, $|u| = i$, $|v| = j+1$,  and $|w| = k$. Owing to the construction of $\cN$, there must exist derivations $0 \xrightarrow{uv} f_i$ in $\cD_m$ and $0' \xrightarrow{vw} f'_j$ in $\cD'_n$, which means $uv \in L_m$ and $vw \in L'_n$.
\item
If $x \in L_m$ and $y \in L'_n$, then by (1), for every $u, v, w$ where $x = uv$ and $y = vw$, $uvw$ is recognized by $\cN$; so $L_m \overlap L_n \subseteq L(\cN)$.
Conversely, if a~word $z$ is recognized by $\cN$, then by (2), $z = uvw$ for some $u, v, w$ where $uv \in L_m$ and $vw \in L_n$; so $L(\cN) \subseteq L_m \overlap L_n$.
Hence $L(\cN) = L_m \overlap L_n$.\qed
\end{enumerate}
\end{proof}

Figure~\ref{fig:NFA-structure} shows the overall structure of the NFA $\cN$, with examples of  transitions of different types.

\begin{figure}[htb]
\unitlength 7pt\scriptsize
\begin{center}\begin{picture}(40,31)(-3,-2)
\gasset{Nh=2,Nw=5,Nmr=2,ELdist=0.2,loopdiam=1.5}

\node[Nh=8,Nw=21,Nmr=4,dash={.5 .5}0](box0)(5,25.5){}
\node(0xs)(5,27){$(0,s')$}\imark(0xs)
\node(0x0)(0,24){$(0,0')$}
\node[Nframe=n](0xdots)(5,24){$\dots$}
\node(0xn)(10,24){$(0,(n$-$1)')$}

\node[Nframe=n](vdots1)(5,20){$\vdots$}

\node[Nh=8,Nw=21,Nmr=4,dash={.5 .5}0](boxf)(5,13.5){}
\node(fxs)(5,15){$(f,s')$}
\node(fx0)(0,12){$(f,0')$}
\node[Nframe=n](fxdots)(5,12){$\dots$}
\node(fxn-1)(10,12){$(f,(n$-$1)')$}

\node[Nh=8,Nw=21,Nmr=4,dash={.5 .5}0](boxt)(31,13.5){}
\node(txs)(31,15){$(t,s')$}
\node(tx0)(26,12){$(t,0')$}
\node[Nframe=n](txdots)(31,12){$\dots$}
\node(txn-1)(36,12){$(t,(n$-$1)')$}

\node[Nframe=n](vdots2)(5,8){$\vdots$}

\node[Nh=8,Nw=21,Nmr=4,dash={.5 .5}0](boxm-1)(5,1.5){}
\node(m-1xs)(5,3){$(m$-$1,s')$}
\node(m-1x0)(0,0){$(m$-$1,0')$}
\node[Nframe=n](m-1xdots)(5,0){$\dots$}
\node[Nw=7](m-1xn-1)(10.5,0){$(m$-$1,(n$-$1)')$}

\drawedge(boxf,boxt){$\varepsilon$~(iv)}
\drawedge[curvedepth=10,ELpos=55](0xs,fxn-1){$a$~(ii)}
\drawedge[curvedepth=4,ELpos=65](0x0,fxn-1){$a$~(iii)}
\drawedge[curvedepth=-11,sxo=-2,exo=-2,syo=-.1,eyo=.1,ELside=r,ELpos=65](0xs,fxs){$a$~(i)}
\drawedge[curvedepth=-4,ELside=r](tx0,txn-1){$a$~(v)}
\end{picture}\end{center}
\caption{The structure of the NFA that accepts the overlap assembly of two regular languages $L_m$ and $L_n'$, with example transitions of every type. Assume that  $D_m$ has the transition  $0 \xrightarrow{a}  f$, that $D_n'$ has the transition $0' \xrightarrow{a} (n-1)'$ and that $f$ is one of the final states of  $D_m$. The first of these two transitions gives rise to $(0, s') \xrightarrow{a} (f, s')$ (type (i)),  while the first and second transition together give rise to $(0, s') \xrightarrow{a} (f, (n-1)')$ (type (ii)) and $(0, 0') \xrightarrow{a} (f, (n-1)')$ (type (iii)).  Since $f$ is final, a~transition $(f, j') \xrightarrow{\varepsilon} (t, j')$ (type (iv)) exists for   all $0\leq j \leq (n-1)$. Lastly,  the second transition gives rise to $(t, 0') \xrightarrow{a} (t, (n-1)')$ (type (v)).   
 }
\label{fig:NFA-structure}
\end{figure}
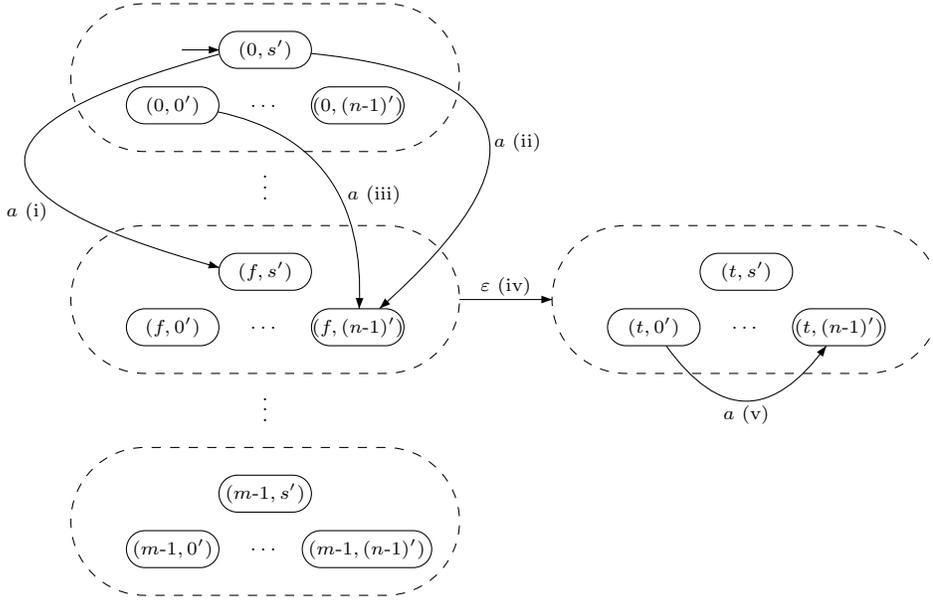

\section{Tight Upper Bound for Overlap Assembly in the General Case}
\label{sec:general}

To establish the state complexity of overlap assembly we need to determinize the $\eps$-NFA $\cN = (R, \Sigma, \eta, r_0, F_\cN)$ defined in Section~\ref{sec:NFA_construction}, and then minimize the resulting DFA.
The first step is to find an~upper bound on the number of subsets $S$ of the set $R$ of states of $\cN$.
We begin by characterizing the reachable subsets of $R$.
They all have the form 
\begin{equation}\label{eq:reach}
S = \{(q,s')\} \cup (\{q\} \times S') \cup (\{t\} \times T'),
\end{equation}
where $q \in Q_m$, $T' \subseteq S' \subseteq Q'_n$ if $q \notin F$, $T' = S' \subseteq Q'_n$ if $q \in F$, and $S'$ is non-empty unless $S=\{ ( 0, s')\}$.
We call $q$ the \emph{selector} of $S$, subset $S' \setminus \{0'\}$ is its \emph{core}, and subset $T'$ is its \emph{subcore}.

We illustrate this using the NFA of Figure~\ref{fig:NFA-example}.
The initial subset is $\{ (0, s')\}$; this has form~(\ref{eq:reach}) with 
$S'=T'=\emptyset$.
From this initial subset we reach by $b$ the subset
$\{(0,s'), (0,0')\} = \{0,s'\} \cup (\{0\} \times \{0'\})$; here $T'=\emptyset$ and $S' =\{0'\}$.
By $a$ we reach 
$\{ (1,s') \} \cup \{(1,1')\} \cup \{(t,1')\} = \{ (1,s') \} \cup  (\{1\} \times \{1'\}) \cup
(\{t\} \times \{1'\})$; here $S'=T'=\{1'\}$.\\

We now proceed to prove the claim about form (1).
\begin{lemma}\label{lem:subsets}
Let $m \geq 2$, $n \ge 1$, and let $\cD$ be the DFA obtained by determinization of the NFA for the overlap assembly $L_m \overlap L_n$.
Every reachable subset of $\cD$ is of the form {\rm (\ref{eq:reach})}.
Moreover, if $q \notin F$, then $S$ cannot be distinguished from $S \cup \{(q,0')\}$.
\end{lemma}
\begin{proof}
First we show that every reachable subset $S \subseteq R$ is of the desired form.
We prove this claim by induction. The initial subset $\{(0,s')\}$ has this form.
Suppose that $S$ has this form, consider a~letter $a \in \Sigma$, and the subset $U = \eta(S,a)$.
Observe that $(\delta_m(q,a),s')$ is the only pair in $U$ containing $s'$, because of the transitions~(i) and because $\cD_m$ is deterministic.
Also, every state $(q,p')$,where $p' \in Q'_n \cup \{s'\}$, is mapped to a~state $(\delta_m(q,a),r') \in \{\delta_m(q,a)\} \times Q'_n$ by the transitions~(ii) and~(iii).
Finally, the states in $\{t\} \times T'$ are mapped only to states from $\{t\} \times Q'_n$ by the transitions~(iv) and~(v).

Note that subsets $S$ with $S' = \emptyset$ are not reachable, unless $S$ is the initial subset $\{(0,s')\}$.

We show that if $S=\{(q,s')\} \cup (\{q\} \times S') \cup (\{t\} \times T')$ is reachable, then $T' \subseteq S'$.
Let $r' \in T'$.
Then there exists a~word $xy$ such that:
$$(0,s) \xrightarrow{x} (q_1,p') \xrightarrow{\eps} (t,p') \xrightarrow{y} (t,r'),$$
where $q_1 \in F$.
We also have:
$$(q_1,p') \xrightarrow{y} (q_2,r').$$
Thus $(q_2,r') \in S$, and so $r' \in S'$. 

We observe that if $q \in F$, then by $\eps$-transitions (transitions~(iv)), every state $(q,r') \in S$ is mapped to $(t,r')$; thus $T' = S'$, which concludes the characterization of reachable subsets.

Finally, we show that if $q \notin F$, then $S$ cannot be distinguished from $S \cup \{(q,0')\}$.
Indeed, let $a \in \Sigma$ be any letter.
Then $\eta((q,0'),a) = \eta((q,s'),a)$ because the transitions~(iii) and~(ii) coincide.
Since $(q,s') \in S$, we have $\eta(S,a) = \eta(S \cup \{(q,0')\},a)$.\qed
\end{proof}

From Lemma~\ref{lem:subsets} two reachable subsets with a~different selector, or a~different core, or a~different subcore are potentially distinguishable.
If two reachable subsets have the same selector, core, and subcore, then they can differ only by state $(q,0')$ if the selector $q$ is not in $F$; thus they cannot be distinguished.
If two reachable subsets have the same selector $q$ that is in $F$, then they cannot differ just by $(q,0')$, as by $\epsilon$-transitions from $(q,0')$ we immediately obtain $(t,0')$.

\begin{theorem}\label{thm:upper_bound}
For $m \geq 2$ and $n \ge 1$, the state complexity of $L_m \overlap L_n$ is at most
$$2 (m-1) 3^{n-1} + 2^n.$$
\end{theorem}
\begin{proof}
Using Lemma~\ref{lem:subsets}, we count the number of potentially reachable and distinguishable subsets
$S = \{(q,s')\} \cup (\{q\} \times S') \cup (\{t\} \times T')$.
\medskip

\noindent\textit{Reachable subsets}:
For every state $q \in Q_m$, we count the number of potentially reachable subsets with selector $q$.
There are 2 cases:
\begin{itemize}
\item If $q$ is non-final, we can choose any non-empty set $S' \subseteq Q_n'$ of cardinality $k$ and any subset $T'$ of $S'$.
The number of ways of doing this is $\sum_{k=1}^n \binom{n}{k} 2^k$.
\item If $q$ is final, again we choose any non-empty set $S'$, but now $T'=S'$ is fixed.
The number of ways of doing this is $2^n - 1$.
\end{itemize}
There is also the initial subset $\{(0,s')\}$ which contributes $1$ to the sum.
In total, this yields:
$$(m-|F|)\cdot \left(\sum_{k=1}^{n} \binom{n}{k} 2^k\right) + |F|\cdot(2^n-1) + 1.$$

\noindent\textit{Distinguishable subsets}:
The above formula gives the number of potentially reachable subsets but overestimates the state complexity because not all subsets are distinguishable.
Recall that by Lemma~\ref{lem:subsets} if the selector $q$ is not in $F$, then $S$ cannot be distinguished from $S \cup \{(q,0')\}$.
Thus we do not need to count subsets $S$ without $0'$, as $S \cup \{(q,0')\}$ is potentially reachable and always equivalent to $S$.
Hence, for a~given $q \in Q_m \setminus F$ we choose $S'$ to be any subset of $Q'_n$ that contains $0'$, and again let $T'$ be any subset of $S'$.
This can be done in $\sum_{k=1}^{n} \binom{n-1}{k-1} 2^k$ ways.
Thus the total number of potentially reachable and distinguishable subsets is at most
$$(m-|F|)\cdot\left(\sum_{k=1}^{n} \binom{n-1}{k-1} 2^{k}\right) + |F|\cdot(2^n-1) + 1.$$
By algebra, we have $\sum_{k=1}^{n} \binom{n-1}{k-1} 2^{k} = 2 \cdot 3^{n-1}$, which is greater than $2^n - 1$; so this formula is maximized when $|F| = 1$, and we conclude that the maximum state complexity of overlap assembly is $2 (m-1) 3^{n-1} + 2^n$.\qed
\end{proof}

\begin{theorem}\label{thm:min_alphabet}
At least $n$ letters are required to meet the bound from Theorem~\ref{thm:upper_bound}.
\end{theorem}
\begin{proof}
Let $q \in F$ be a~final state of $\cD_m$.
For each $p' \in Q'_n$ we consider the subset
$$T_{p'} = \{(q,s'),(q,p'),(t,p')\}.$$
If the upper bound is met, then, in particular, all subsets $S$ with $q \in F$ must be reachable in view of
Lemma~\ref{lem:subsets}.
These subsets were counted in the upper bound, and there are no other subsets of reachable form that could be equivalent to them when the upper bound is met.
Hence, in particular, all subsets $T_{p'}$ must be reachable.

Suppose that $T_{p'}$ is reachable by a~word $w_{p'} a_{p'}$, for some letter $a_{p'}$.
Note that $(q,p')$ is the only one of the three states in $T_{p'}$ that can be reached by transitions~(ii) of the NFA.
Consider $\eta(r_0,w_{p'})$;
it must contain $(r,s')$ for some $r \in Q_m$, because by Lemma~\ref{lem:subsets} every reachable subset has exactly one such pair. 
Thus, $(r,s')$ must be mapped by transitions~(ii) induced by $a_{p'}$ to $(q,p')$.
Therefore, $\delta'_n(0',a_{p'}) = p'$, which proves that $a_{p'}$ are different for every $p'$.\qed
\end{proof}

We define the witness DFAs for $m,n \ge 2$.
Let $\Sigma = \{a_0,\ldots,a_{n-1}\}$.
Let $\cW_m=(Q_m, \Sigma, \delta_m, 0, F)$ be defined as follows:
\begin{itemize}
\item $F = \{0\}$;
\item $a_i\colon \mathbf{1}_m$ for $i \in \{0,2,\ldots,n-1\}$, where $\mathbf{1}_m$ is the identity transformation on $Q_m$;
\item $a_1\colon (0,1,\ldots,m-1)$ is a~cyclic permutation of $Q_m$.
\end{itemize}
Let $\cW'_n=(Q'_n, \Sigma, \delta'_n, 0', F')$ be defined as follows:
\begin{itemize}
\item $F = \{(n-1)'\}$;
\item $a_0\colon (Q'_n \to 0')$ maps all the states of $Q'_n$ to $0'$;
\item $a_i\colon (1',2',3',\ldots,(i-1)',0',i',\ldots,(n-1)')$ for $i \in \{1,\ldots,n-1\}$.
 Here $a_i$ permutes the states of $Q'_n$, mapping $1'$ to $2'$, $2'$ to $3'$, etc., then $(i-1)'$ to $0'$, $0'$ to $i'$, and then $i'$ to $(i+1)'$, etc., and $(n-1)'$ to $1'$.
\end{itemize}
 The transitions of these DFAs with $m=3$ and $n=4$ states are illustrated in Figure~\ref{fig:witness3x4}.
Let $L_m$ and $L'_n$ be the languages of $\cW_m$ and $\cW'_n$, respectively.

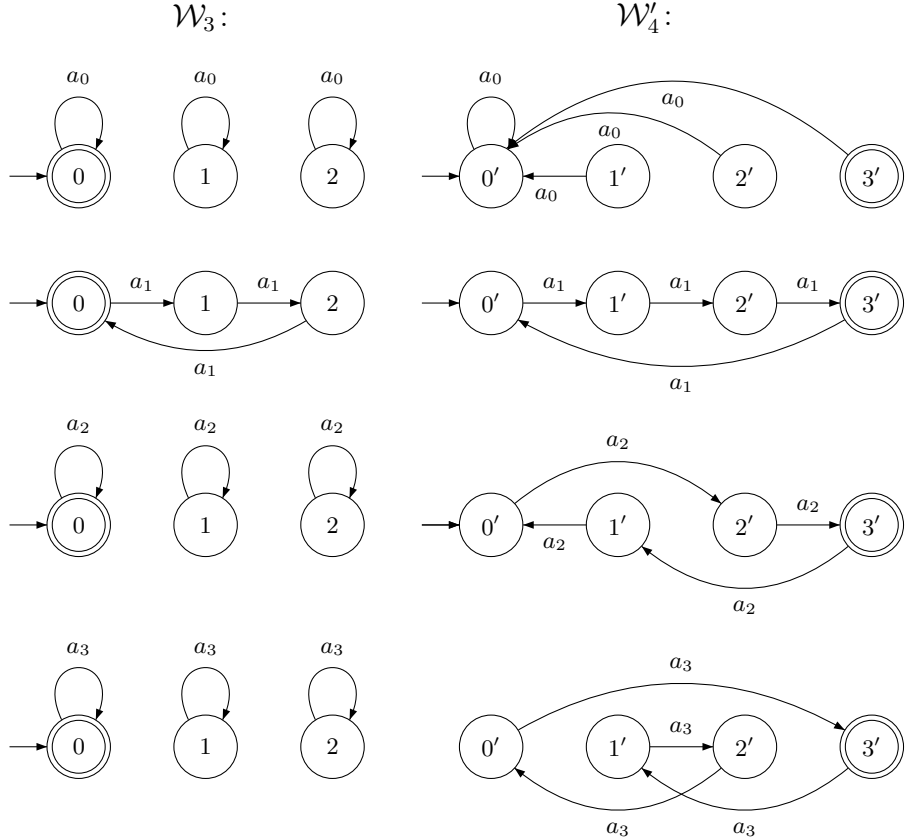
\begin{figure}[th]
\unitlength 12pt
\begin{center}\begin{picture}(26,27)(-2,-3)
\gasset{Nh=2,Nw=2,Nmr=1.25,ELdist=0.4,loopdiam=1.5}
\node[Nframe=n](name1)(4,23){\large$\cW_3\colon$}
\node[Nframe=n](name2)(18,23){\large$\cW'_4\colon$}
\node(0a0)(0,18){$0$}\rmark(0a0)\imark(0a0)
\node(1a0)(4,18){$1$}
\node(2a0)(8,18){$2$}
\drawloop(0a0){$a_0$}
\drawloop(1a0){$a_0$}
\drawloop(2a0){$a_0$}
\node(0a1)(0,14){$0$}\rmark(0a1)\imark(0a1)
\node(1a1)(4,14){$1$}
\node(2a1)(8,14){$2$}
\drawedge(0a1,1a1){$a_1$}
\drawedge(1a1,2a1){$a_1$}
\drawedge[curvedepth=1.5](2a1,0a1){$a_1$}
\node(0a2)(0,7){$0$}\rmark(0a2)\imark(0a2)
\node(1a2)(4,7){$1$}
\node(2a2)(8,7){$2$}
\drawloop(0a2){$a_2$}
\drawloop(1a2){$a_2$}
\drawloop(2a2){$a_2$}
\node(0a3)(0,0){$0$}\rmark(0a3)\imark(0a3)
\node(1a3)(4,0){$1$}
\node(2a3)(8,0){$2$}
\drawloop(0a3){$a_3$}
\drawloop(1a3){$a_3$}
\drawloop(2a3){$a_3$}

\node(0'a0)(13,18){$0'$}\imark(0'a0)
\node(1'a0)(17,18){$1'$}
\node(2'a0)(21,18){$2'$}
\node(3'a0)(25,18){$3'$}\rmark(3'a0)
\drawloop[loopangle=90](0'a0){$a_0$}
\drawedge[exo=-.5](1'a0,0'a0){$a_0$}
\drawedge[curvedepth=-2,exo=-.5](2'a0,0'a0){$a_0$}
\drawedge[curvedepth=-3,exo=-.5](3'a0,0'a0){$a_0$}

\node(0'a1)(13,14){$0'$}\imark(0'a1)
\node(1'a1)(17,14){$1'$}
\node(2'a1)(21,14){$2'$}
\node(3'a1)(25,14){$3'$}\rmark(3'a1)
\drawedge(0'a1,1'a1){$a_1$}
\drawedge(1'a1,2'a1){$a_1$}
\drawedge(2'a1,3'a1){$a_1$}
\drawedge[curvedepth=2](3'a1,0'a1){$a_1$}

\node(0'a2)(13,7){$0'$}\imark(0'a2)
\node(1'a2)(17,7){$1'$}
\node(2'a2)(21,7){$2'$}
\node(3'a2)(25,7){$3'$}\rmark(3'a2)
\drawedge[curvedepth=2](3'a2,1'a2){$a_2$}
\drawedge(1'a2,0'a2){$a_2$}
\drawedge[curvedepth=2](0'a2,2'a2){$a_2$}
\drawedge(2'a2,3'a2){$a_2$}

\node(0'a3)(13,0){$0'$}\imark(0'a2)
\node(1'a3)(17,0){$1'$}
\node(2'a3)(21,0){$2'$}
\node(3'a3)(25,0){$3'$}\rmark(3'a3)
\drawedge[curvedepth=2](3'a3,1'a3){$a_3$}
\drawedge(1'a3,2'a3){$a_3$}
\drawedge[curvedepth=2](2'a3,0'a3){$a_3$}
\drawedge[curvedepth=2](0'a3,3'a3){$a_3$}
\end{picture}\end{center}
\caption{The actions of the letters in $\cW_3$ and $\cW'_4$.}\label{fig:witness3x4}
\end{figure}

By a~\emph{cyclic shift} of a~core subset $S' \subseteq \{1',\ldots,(n-1)'\}$ we understand any subset obtained by shifting the states along the cycle $(1',\ldots,(n-1)')$, $i$ positions clockwise, i.e.,\ the subset $\{(((p-1+i) \bmod (n-1))+1)' \mid p' \in S'\}$ for any $i \ge 0$.
The \emph{next} and \emph{previous} cyclic shifts correspond to $i=1$ and $i=n-2$, respectively.

The transitions of letters $a_1,a_2,\ldots,a_{n-1}$ produce next cyclic shifts of the states in $\{1',\ldots,(n-1)'\}$, with the exception that state $0'$ replaces one of the states in the cycle.
The idea behind the witness is that we can add an~arbitrary state to the core using these letters and produce arbitrary cyclic shifts as well, as will be shown later.
Letter $a_0$ plays an~important role of reset, which is necessary to reach small subsets.
The main difficulty is that $a_1$ shares both roles of producing cyclic shifts and switching the selector.

\begin{theorem}
\label{thm:meets_upper_bound}
For $m \ge 2$ and $n \ge 3$, $L_m \overlap L'_n$ meets the upper bound.
\end{theorem}
\begin{proof}
\noindent\textit{Reachability}:
It is enough to show that all subsets $S$ from Lemma~\ref{lem:subsets} are reachable, with the exception that if $q \notin F$ then it suffices to show reachability of either $S \setminus \{(q,0')\}$ or $S \cup \{(q,0')\}$.

\noindent$\bullet$ First we show that for all subsets
$$S = \{(q,s')\} \cup (\{q\} \times S'),$$
where $q \in Q_m \setminus \{0\}$ and $\emptyset \neq S' \subseteq Q_n' \setminus \{0'\}$, either $S \setminus \{(q,0')\}$ or $S \cup \{(q,0')\}$ is reachable.
These subsets have core $S'$ and an~empty subcore.

We prove this by induction on the size $|S'|$ of the core.
For $|S'|=0$, apply $a_1^q a_0$ to $(0,s')$; this yields $\{(q,s'),(q,0')\}$.

Consider $|S'|=1$.
If $q=1$, then we just use $a_1$, which yields $\{(1,s'),(1,1')\}$.
To meet the other subsets $\{(1,s'),(1,p')\}$ for $p \ge 2$, from $\{(1,s'),(1,1')\}$ we use $a_0 a_p$.
For $q \ge 2$, we use $a_1^{q-1} a_0 a_1$, which yields $\{(q,s'),(q,1')\}$.
Then to meet the other subsets $\{(q,s'),(q,p')\}$ for $p \ge 2$, from $\{(q,s'),(q,1')\}$ we also use $a_0 a_p$.

Consider $|S'|\ge 2$ and assume the induction hypothesis for subsets $S$ with a~smaller core.
Since $S'$ contains at least two states different from $0'$, there is a~state $p' \in S' \setminus \{1'\}$.
Let $X'$ be the previous cyclic shift of $S' \setminus \{p'\}$.
Since $p' \notin S' \setminus \{p'\}$, $X'$ does not contain $(p-1)'$, but this is its only difference from the previous cyclic shift of $S'$.
By the inductive assumption, $\{(q,s')\} \cup (\{q\} \times X')$ is reachable.
We apply $a_p$ to this subset, which maps $X'$ to its next cyclic shift, and also $(q,s')$ to $(q,p')$, which yields $\{(q,s'\} \cup (\{q\} \times S')$.

\noindent$\bullet$ Now we show reachability of subsets
$$S = \{(0,s')\} \cup (\{0\} \times S') \cup (\{t\} \times S'),$$
where $\emptyset \neq S' \subseteq Q_n'$.
These are all potentially reachable subsets with selector $0$.

First consider the case $0' \notin S'$.
For $\{(m-1,s'),(m-1,1')\}$ we apply $a_0 a_1$, which yields $\{(0,s'),(0,1'),(t,1')\}$.
Then we continue the induction on $|S'|$ as before when $|S'| \ge 2$, with just $\{t\} \times S'$ added to the subsets.

Now consider the case $0' \in S'$.
The case $S'=\{0'\}$ is easily covered by applying $a_0$ to $\{(0,s'),(0,1'),(t,1')\}$.
If $S' = \{0',1'\}$, then from $\{(m-1,s'),(m-1,(n-1)')\}$ we apply $a_1$ and get $\{(0,s'),(0,0'),(0,1'),(t,0'),(t,1')\}$ as desired.
Let $S' \neq \{0',1'\}$.
We already know that $\{(0,s')\} \cup (\{0,t\} \times X')$ is reachable, where $X'$ is the previous cyclic shift of $S' \setminus \{0'\}$.
Since $|S'|\ge 2$ and $S' \neq \{0',1'\}$, there is a~$p' \in S' \setminus \{1'\}$.
We apply $a_p$ to $\{(0,s')\} \cup (\{0,t\} \times X')$.
We have $X' \setminus \{(p-1)'\}$ mapped to $S' \setminus \{p'\}$ and $(p-1)'$ mapped to $0'$, which gives $(\{0\} \times (S' \cup \{0'\} \setminus \{p'\})$ by transitions~(iii), and $(0,p')$ is added by transitions~(ii).
Thus, after completing by $\eps$-transitions this yields $\{(0,s')\} \cup (\{0,t\} \times S')$.

\noindent$\bullet$ Finally, we show that for all subsets
$$S = \{(q,s')\} \cup (\{q\} \times S') \cup (\{t\} \times T'),$$
where $q \neq 0$ and $\emptyset \neq T' \subseteq S' \subseteq Q'_n$,
either $S \setminus \{(q,0')\}$ or $S \cup \{(q,0')\}$ is reachable.

Consider the special case $S'=T'=\{0'\}$.
We reach it from $\{(0,s'),(0,0'),(t,0')\}$ by applying $a_1^q a_0$.
For the rest, assume that $S' \setminus \{0'\}$ is non-empty.

We need an~auxiliary argument that from $\{(0,s')\}$ we can reach
a subset with selector $q$, core $S'$, and an~empty subcore,
using a~word from $\{a_1,a_2,\ldots,a_{n-1}\}^*$ (any word without $a_0$).
We prove this by induction on the core size $|S' \setminus \{0'\}|$.
For $|S' \setminus \{0'\}|=1$, at the beginning we use $a_1$, which yields $\{(1,s'),(1,1')\}$.
Now we can reach $\{(1,s'),(1,0'),(1,p')\}$ for any $p' \in \{2',\ldots,(n-1)'\}$ by using $a_2 a_3 \dots a_p$.
Then, from $\{(1,s'),(1,0'),(1,(n-1)')\}$ we reach $\{(2,s'),(2,0'),(2,1')\}$,
and it remains to repeat the argument to reach every remaining subset of the form $\{(q,s'),(q,0'),(q,p')\}$ for $q \in Q_m \setminus \{0,1\}$ and $p' \in Q'_n \setminus \{0'\}$.
For $|S' \setminus \{0'\}|\ge 2$ we follow the first part of the reachability argument as before, but we reach either $\{(q,s')\} \cup (\{q\} \times (S' \setminus \{0'\})$ or $\{(q,s')\} \cup (\{q\} \times (S' \cup \{0'\}))$, instead of just the former.
Let  $w \in \{a_1,a_2,\ldots,a_{n-1}\}^*$ be a~word that reaches either $\{(q,s')\} \cup (\{q\} \times (S' \setminus \{0'\})$ or $\{(q,s')\} \cup (\{q\} \times (S' \cup \{0'\}))$.

Suppose that we start from the subset
$$S_0 = \{(0,s')\} \cup (\{0,t\} \times T'_0),$$
where $T'_0$ is some subset such that $\emptyset \neq T'_0 \subseteq Q'_n$.
We already know that for every $T'_0$, subset $S_0$ is reachable.
After applying $a_1 w$, we reach either
$$S_q = \{(q,s')\} \cup (\{q\} \times (S' \cup T'_q \setminus \{0'\})) \cup (\{t\} \times T'_q),$$
or $S_q \cup \{(q,0')\}$, where $T'_q$ is obtained by applying some permutation $\pi$ of $Q'_n$ to $T'_0$.
This is because $\{(0,s')\}$ is mapped by $a_1 w$ to $\{(q,s')\} \cup (\{q\} \times (S' \setminus \{0'\})$ or $\{(q,s')\} \cup (\{q\} \times (S' \cup \{0'\}))$, word $a_1 w$ acts as a~permutation on $(\{t\} \times Q'_q)$, and $\{0\} \times T'_0$ is mapped to $(\{q\} \times T'_q)$.
Note that $a_1 w$ does not depend on $T'_0$, so we can choose $T'_0$ arbitrarily.
Let $T'_0=\pi^{-1}(T')$, so $\pi(T'_0) = T'$.
We obtain either
$$S_q = \{(q,s')\} \cup (\{q\} \times ((S' \setminus \{0'\}) \cup T') \cup (\{t\} \times T'),$$
or
$$S_q = \{(q,s')\} \cup (\{q\} \times ((S' \cup \{0'\}) \cup T') \cup (\{t\} \times T').$$
Recall that $T' \subseteq S'$ and if $0' \in T$, then also $0' \in S'$; hence $(S' \setminus \{0'\}) \cup T'$ is either $S'$ or $S' \setminus \{0'\}$, and $(S' \cup \{0'\}) \cup T' = S' \cup \{0'\}$.
Thus, $S_q$ is either $S \setminus \{(q,0')\}$ or $S \cup \{(q,0')\}$.
\medskip

\noindent\textit{Distinguishability}:
Consider two reachable subsets
$$S_1 = \{(q_1,s')\} \cup (\{q_1\} \times S'_1) \cup (\{t\} \times T'_1),$$
and
$$S_2 = \{(q_2,s')\} \cup (\{q_2\} \times S'_2) \cup (\{t\} \times T'_2),$$
with different selectors, different cores, or different subcores.
Thus we have $q_1 \neq q_2$, or $T'_1 \neq T'_2$, or $(S'_1 \setminus \{(q_1,0')\}) \neq (S'_2 \setminus \{(q_2,0'\})$.
These are precisely all\ the reachable and potentially distinguishable subsets in view of Lemma~\ref{lem:subsets}.
Note that the initial subset also has this form, where $q_1=0$ and $S'_1$ and $T'_1$ are empty.

If $q_1 \neq q_2$, then without loss of generality let $q_1 < q_2$.
We apply $a_1^{m-q_2} a_0 a_{n-1}^2$.
For $S_1$, first $a_1^{m-q_2} a_0$ maps it to a~subset $\{(q,s'),(0,s')\}$ or $\{(q,s'),(q,0'),(t,0')\}$ (if $T'_1$ is non-empty) for some $q \neq 0$.
Then $a_{n-1}^2$ results in a~subset that from the states from $(\{t\} \times Q'_n)$ contains at most $(t,1')$, which is not final.
On the other hand, $S_2$ by $a_1^{m-q_2} a_0$ is mapped to $\{(0,s'),(0,0'),(t,0')\}$.
Then $a_{n_1}^2$ yields $\{(0,s'),(0,0'),(t,1'),(t,(n-1)')\}$, where $(t,(n-1)')$ is final.

So suppose that $q_1 = q_2$.
If $q_1 \neq 0$ and $T'_1 \neq T'_2$, then we apply $a_{n-1}^i$ for a~suitable $i \ge 0$.
Since $a_{n-1}$ acts cyclically on all states $(\{t\} \times Q'_n)$ and no other states from the subsets are mapped to $(\{t\} \times Q'_n)$, we can repeat the cycle so that exactly one of $\eta(\{t\} \times T'_1,a_{n-1}^i)$ and $\eta(\{t\} \times T'_2,a_{n-1}^i)$ contains the final state $(t,(n-1)')$.
If $q_1 = 0$ and $T'_1 \neq T'_2$, then also $S'_1 \neq S'_2$, so it remains to cover this case.

Suppose that $S'_1 \neq S'_2$.
If $q_1 = q_2 = 0$, then also $T'_1 \neq T'_2$.
We apply $a_1$, which maps $S_1$ to the subset
$$\{(1,s')\} \cup (\{1\} \times (\delta_m(S'_1,a_1) \cup \{2'\})) \cup (\{t\} \times \delta'_n(T'_1,a_1)),$$
and analogously $S_2$.
Since $T'_1 \neq T'_2$ and $a_1$ acts cyclically on $Q'_n$, we have $\delta'_n(T'_1,a_1) \neq \delta'_n(T'_1,a_1)$.
The case of these subsets has been already covered in the previous paragraph.

There remains the case where $T'_1 = T'_2$, $S'_1 \neq S'_2$, $q_1 = q_2 \neq 0$.
We follow the induction on the selector $q_1$ starting with $q_1=m-1$ and decreasing it.
We will show for $q_1=m-1$ that we can reach subsets with selector $0$ that still have different cores.
We have already shown in the previous paragraph that the subsets with selector $0$ and different cores can be distinguished.
For $q_1<m-1$ we will show that we can reach subsets with the same property but with selector $q_1+1$, which will follow by the inductive assumption.
So let $p$ be the largest index such that, without loss of generality, $p' \in S'_1$ and $p' \notin S'_2$.
Note that $p \neq 0$, because then the subsets cannot be distinguished.
If $p < n-1$, then we apply $a_1$, which yields subsets with the desired property.
If $p = n-1$, then we first apply $a_2$, which yields the subset with $p'=1'$, and then we can apply $a_1$ as before.\qed
\end{proof}

\section{Unary Alphabet}
\label{sec:unary}

In this section, we consider overlap assembly of languages over a~one-letter alphabet.
First note that if the longest word that is in a~unary language $L$ is of length $n$, then the state complexity of $L$ is exactly $n+2$.
Similarly, if the longest word that is \emph{not} in a~unary language $L$ is of length $n$, then the state complexity of $L$ is exactly $n+2$~\cite{YZS94}.

\begin{theorem}
\label{thm:unary}
Let $m,n \geq 1$, and let $L_m$ and $L_n$ be two unary languages of state complexities $m$ and $n$, respectively. The state complexity of $L_m \overlap L_n$ is at most $m+n$, and this bound is met by
$L_m = \{a^{mk+n-1} \mid k \in \mathbb{Z}, mk+n-1 \geq 0\} $ and
$L_n = \{a^{nk+m-1} \mid k \in \mathbb{Z}, nk+m-1 \geq 0\}.$
\end{theorem}
\begin{proof}
We consider three cases:
\begin{description}
\item{\bf Two infinite languages}

Since languages $L_m$ and $L_n$ are regular and infinite, there are some $i,j \leq m$ and $i', j' \leq n$ such that $L_m \supseteq \{a^{ik+j} \mid k \geq 0\}$ and $L_n \supseteq \{a^{i'k' + j'} \mid k' \geq 0\}$.

Let $t \geq m+n-1$; we show that $a^t \in L_m \overlap L_n$.
Choose $k$ and $k'$ to be the maximum integers such that $ik + j \leq t$ and $i'k' + j' \leq t$. The longest word in $a^{ik + j} \overlap a^{i'k' + j'}$ is $a^{(ik+j) + (i'k'+j') - 1}$. By definition of $k$, we have $ik+j+i > t$; so $ik+j \geq t-i+1$. 
Similarly, $i'k'+j' \geq t-i'+1$. However, 
\begin{gather*}
(ik + j) + (i'k' + j') - 1 \geq (t-i+1) + (t-i'+1) - 1 \\
= 2t - i - i' + 1  \geq 2t - m - n + 1 \geq t.
\end{gather*}
Therefore for any $t \geq m+n-1$, $a^t \in a^{ik + j} \overlap a^{i'k' + j'}$. The longest word that might not be in $L_m \overlap L_n$ is $a^{m+n-2}$, and so  the state complexity of $L_m \overlap L_n$ is at most $m+n$.

Next, we prove that the bound is met by the languages given in the theorem.
Since we showed that $L_m \overlap L_n$ contains all $a^t$ with $t \geq m+n-1$, it is sufficient to show that $a^{m+n-2}$ is not in  $L_m \overlap L_n$.
Note that $a^{m+n-1}$ is in both $L_m$ and $L_n$, and we cannot obtain $a^{m+n-2}$ if either word in $L_m$ or $L_n$ has length $\ge m+n-1$. Therefore we only need to consider the next longest words, which are $a^{n-1} \in L_m$ and $a^{m-1} \in L_n$.
Since the longest word in $a^{n-1} \overlap a^{m-1}$ is $a^{m+n-3}$, we have $a^{m+n-2} \notin L_m \overlap L_n$.
Therefore the state complexity is $m+n$.

\item{\bf Two finite languages}

Now the longest word in $L_m$ is $a^{m-2}$ and the longest word in $L_n$ is $a^{n-2}$. Therefore the longest word in $L_m \overlap L_n$ is $a^{m+n-5}$. 
Hence the state complexity of $L_m \overlap L_n$ is exactly $m+n-3$.

\item{\bf An~infinite language and a~finite one}

We prove the following claim: Let $m,n \ge 1$, let $L_m$ be an~infinite unary language, and let $L_n$ be a~finite unary language.
If $m \leq n-2$, then the state complexity of $L_m \overlap L_n$ is at most $n-1$.
Otherwise, it is at most $m+n-2$.

We consider the following two cases:
	\begin{enumerate}
	\item
	$m\le n-2$
	
	We show that for  $t \geq n-2$,  $a^t \in L_m \overlap L_n$. By definition of $L_m$, 	there exists $a^s \in L_m$ with $s \leq t$ and $t-s \leq m-1 	\leq n-3$. 
	Hence  $a^t \in a^s \overlap a^{n-2}$ and so $a^t \in L_m \overlap L_n$.
	Therefore the state 	complexity of $L_m \overlap L_n$ is at most $n-1$. 
	\item
	$m> n-2$
	
	We show that there is $i \ge 1$  such that for all $t \ge n+m-2$ we have $a^t \in L_m \overlap L_n$ if and only if $a^{t-i} \in L_m \overlap L_n$.
	This proves that the quotients of $a^t$ and of $a^{t-i}$ are equal, so there exists a~unary DFA (not necessarily minimal) recognizing $L_m \overlap L_n$ with a~cycle of length $i$ and $n+m-2$ states.
	
	Let $i$ be the length of the cycle in a~minimal DFA of $L_m$.
	Then $i \le m$ and $m-i$ is the number of states in the initial path in this DFA.
	Since $L_n$ is finite, $a^{n-2}$ is its longest word.

	First assume that $a^t \in L_m \overlap L_n$.
	Then there are $a^{ik+x} \in L_m$ and $a^y \in L_n$ such that
	$k\ge 0$, $x \le m-1$, $y \le n-2$, and $\max\{ik+x,y\} \le t \le ik+x+y-1$.
	Because $x+y-1 \le m+n-4$ and $t \ge n+m-2$, it must be that $k \ge 1$.
	Then $a^{i(k-1)+x} \in L_m$.
	We have $t-i \ge (n+m-2)-m \ge n-2 \ge y$ and $i(k-1)+x \le t-i$, thus $\max\{i(k-1)+x,y\} \le t-i$.
	Also, from $t \le ik+x+y-1$ we have $t-i \le i(k-1)+x+y-1$.
	Therefore, $a^{ i(k-1)+x} \in L_m$ and $a^y \in L_n$ form $a^{t-i} \in L_m \overlap L_n$.

	Now assume that $a^{t-i} \in L_m \overlap L_n$.
	Since $a^{t-i} \in L_m \overlap L_n$, there are $a^{ik+x} \in L_m$ and $a^y \in L_n$ such that
	$k\ge 0$, $x \le m-1$, $y \le n-2$, and $\max\{ik+x,y\} \le t-i \le ik+x+y-1$.
	If $x \le m-i-1$, then $x+y-1 \le (m-i-1)+(n-2)-1 = m+n-i-4$ but $t-i \ge n+m-2-i$, which yields a~contradiction.
	If $x \ge m-i$, then $a^{ik+x}$ is accepted in a~state in the cycle of the DFA of $L_m$.
	Thus $a^{i(k+1)+x} \in L_m$ and, together with $a^y$, $a^{i(k+1)+x}$ forms $a^t \in L_m \overlap L_n$.
	Hence the state complexity of $L_m \overlap L_n$ is at most $m+n-2$.

	\end{enumerate}
\end{description}
In summary, the largest upper bound occurs if both languages are infinite, and the theorem holds.\qed
\end{proof}

\section{Binary Alphabet}\label{sec:binary}

We define the following binary DFAs for $m,n \ge 2$.
Let $\Sigma = \{a_0,a_1\}$.
Let $\cB_m(Q_m, \Sigma, \delta_m, 0, F)$ be defined as follows:
\begin{itemize}
\item $F = \{0\}$;
\item $a_0\colon \mathbf{1}_m$;
\item $a_1\colon (0,1,\ldots,m-1)$.
\end{itemize}
Let $\cB'_n(Q'_n, \Sigma, \delta'_n, 0', F')$ be defined as follows:
\begin{itemize}
\item $F = \{(n-1)'\}$;
\item $a_0\colon (1',\ldots,(n-1)')$;
\item $a_1\colon (0',1',\ldots,(n-1)')$.
\end{itemize}

\begin{figure}[thb]
\unitlength 14pt\small
\begin{center}\begin{picture}(18,14)(-4,-2)
\gasset{Nh=2,Nw=2,Nmr=1.25,ELdist=0.4,loopdiam=1.5}
\node[Nframe=n](name)(-2,12){\large$\mathcal{B}_m\colon$}
\node(0)(0,8){$0$}\rmark(0)
\node(1)(4,8){$1$}
\node[Nframe=n](dots)(8,8){\dots}
\node(m-1)(12,8){$m$-$1$}
\drawloop(0){$a_0$}
\drawloop(1){$a_0$}
\drawloop(dots){$a_0$}
\drawloop(m-1){$a_0$}
\drawedge(0,1){$a_1$}
\drawedge(1,dots){$a_1$}
\drawedge(dots,m-1){$a_1$}
\drawedge[curvedepth=2](m-1,0){$a_1$}

\node[Nframe=n](name')(-2,4){\large$\mathcal{B}'_n\colon$}
\node(0')(0,2){$0'$}
\node(1')(4,2){$1'$}
\node[Nframe=n](dots')(8,2){\dots}
\node(n-1')(12,2){$(n$-$1)'$}\rmark(n-1')
\drawedge(0',1'){$a_1$}
\drawedge(1',dots'){$a_0,a_1$}
\drawedge(dots',n-1'){$a_0,a_1$}
\drawedge[curvedepth=3.5](n-1',0'){$a_1$}
\drawloop(0'){$a_0$}
\drawedge[curvedepth=2](n-1',1'){$a_0$}
\end{picture}\end{center}
\caption{ Binary automata $\mathcal{B}_m$ and $\mathcal{B}'_n$ such that  
$L(\mathcal{B}_m ) \odot L(\mathcal{B}'_n) $  has exponential state complexity.}
\end{figure}
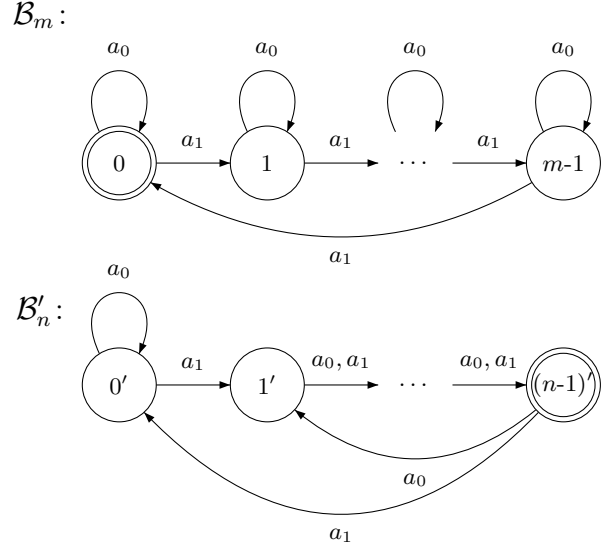

\begin{theorem}
For $m \ge 2$ and $n \ge 3$, the state complexity of $L(\mathcal{B}_m) \overlap L(\mathcal{B}'_n)$ is at least $m(2^{n-1}-2)+2$.
\end{theorem}
\begin{proof}
The proof is based on   ideas  similar to those in the proof of Theorem~\ref{thm:meets_upper_bound}. 

\medskip

\noindent\textit{Reachability}:
We show that for each selector $q \in Q_m$ and each core $\emptyset \neq S' \subseteq Q'_n \setminus \{0'\}$, there exists a~reachable subset  $S$ with some subcore, that is:
$$S = \{(q,s')\} \cup (\{q\} \times (S' \cup \{0'\})) \cup \{\{t\} \times T'\},$$
for some subcore $T' \subseteq S' \cup \{0'\}$.

First, we show that we can reach a~subset of that form but for some selector $p \in Q_m$ that is not necessarily $q$.
We prove this by induction on $|S'|$.
For $S'=\{r'\}$, we apply $a_1 a_0^{r-1}$, which yields $\{(1,s'),(1,0'),(1,r')\}$.
Let $|S'| \ge 2$ and assume that the claim holds for smaller subsets $S'$.
Let $r' \in S'$ be a~state and let $X' = S' \setminus \{r'\}$,
By assumption we can reach
$$X = \{(p,s')\} \cup (\{p\} \times (X' \cup \{0'\})) \cup \{\{t\} \times Y'\},$$
for some $Y' \subseteq X'$.
We apply $a_0^{m-1-r} a_1 a_0^{r-1}$ for $X$.
This first maps $X'$ to its cyclic shift without state $(m-1)'$, then state $1'$ is added by $a_1$ and the selector is changed, and we again cyclically shift to get $X'$.
Finally, we apply $a_0^{n-1}$ to ensure that $(q,0')$ is present; this yields the desired subset $S$.

Now, to change the selector from $p$ to $q$ we use the same technique.
It is enough to show that from a~subset with selector $p$ we can reach a~subset with the selector $(p+1) \bmod m$ and the same core $S'$.
We choose a~state $r' \in S'$, and then use $a_0^{m-1-r} a_1 a_0^{r-1}$.
This first changes the core so that $(m-1)'$ is there, then the selector is changed by $a_1$, and the core is cyclically shifted back to $S'$.

\medskip

\noindent\textit{Distinguishability}:
We will show that all the  subsets above such that $S' \neq Q'_n \setminus \{0'\}$ together with the initial subset and one of the subsets with $S'=Q'_n \setminus \{0'\}$ are pairwise distinguishable.
The number of non-empty and not full cores $S'$ is $2^{n-1}-2$, which together with the $m$ choices for the selector $q$ yields $m(2^{n-1}-2)$.
Adding the initial subset and the subset with full $S'$ yields the desired formula.

Without loss of generality, let
\begin{align*}
S_1 =\ & \{(q_1,s')\} \cup (\{q_1\} \times (S_1' \cup \{0'\})) \cup \{\{t\} \times T_1'\},\\
S_2 =\ & \{(q_2,s')\} \cup (\{q_2\} \times (S_2' \cup \{0'\})) \cup \{\{t\} \times T_2'\},
\end{align*}
be such that  $\emptyset \neq S'_1 \subsetneq Q'_n \setminus \{0'\}$,
$\emptyset \neq S'_2 \subseteq Q'_n \setminus \{0'\}$,
$T_1',T_2' \subseteq Q'_n$,
and $S'_1 \neq S'_2$ or $q_1 \neq q_2$.
Moreover, we can assume that $|S'_1| \le |S'_2|$.

First consider the case $q_1 \neq q_2$.
Let $r'$ be such that $r' \in S'_1$.
As before, by applying $a_0^{n-1-r} a_1 a_0^{r-1}$, from $S_1$  we reach a~subset with selector $(q_1+1) \bmod m$ and the same core $S'_1$.
Similarly $S_2$ is mapped to a~subset with selector $(q_2+1) \bmod m$.
We repeat this procedure until $S_2$ is mapped to a~subset with selector $(m-1,s')$, that is, for $S_1$ and $S_2$ we apply $(a_0^{n-1-r} a_1 a_0^{r-1})^{m-1-r}$.
Since $q_1 \neq q_2$, the first subset obtained from $S_1$ has selector $q \neq m-1$.
Now let $p' \in Q'_n \setminus (S'_1 \cup \{0'\})$.
We apply $a_0^{n-1-p}$, which  causes $(n-1)'$ to be absent from the core of the first subset.
Since a~subcore is always a~subset of the core with $(0,t')$ added, $(n-1)'$ is also  absent from the subcore of the first subset.
We apply $a_1$ and obtain:
\begin{align*}
X_1 =\ & \{(q+1,s')\} \cup (\{q+1\} \times Y'_1) \cup \{\{t\} \times Z'_1\},\\
X_2 =\ & \{(0,s')\} \cup (\{0\} \times Y'_2) \cup \{\{t\} \times Z'_2\},
\end{align*}
for some $Z'_1 \subseteq Y'_1 \subseteq Q'_n$ and $Z'_2 \subseteq Y'_2 \subseteq Q'_n$.
Since $(n-1)'$ was not in the subcore of the first subset and $q+1 \neq 0$, we have $0' \notin Z'_1$.
We apply $a_0^{n-1}$.
Since $0' \notin Z'_1$ and $q+1 \neq 0$, from $X_1$ we obtain a~subset that does not have final state $(t,(n-1)')$.
On the other hand, from $X_2$ state $(0,s')$ is mapped by $a_0$ to $(0,1')$ and then by an~$\eps$-transition to $(t,1')$.
This is then mapped to final state $(t,(n-1)')$ by $a_0^{n-2}$.

Now consider the case $q_1 = q_2$ and $S'_1 \neq S'_2$.
Since $S'_2$ is not a~subset of $S'_1$, there is a~state $p'$ such that $p' \notin S'_1$ and $p' \in S'_2$.
Let $r' \in S'_1$.
We apply $a_0^{m-1-r} a_1 a_0^{r-1}$ as before, which changes the selector to $(q_1+1) \bmod m$, but does not change the core $S'_1$ of the first subset.
We repeat this until selector $0$ is reached.
Then we still have $p' \notin S'_1$ but $p' \in Y'_2$, where $Y'_2$ is the core of the second subset.
We apply $a_0^{n-1-p}$.
Then the first subset does not have final state $(t,(n-1)')$, but the second one does.

Finally, we need to distinguish the initial subset from the other subsets.
For the initial subset, we observe that applying either $a_0 a_1 a_0^{n-1}$ or $a_1 a_0^{n-1}$ results in $\{(1,s'),(1,0'),(1,1')\}$.
On the other hand, every other subset that we have to consider has a~non-empty core $S'_2$.
If $S'_2 = \{(n-1)'\}$ then we apply $a_0 a_1 a_0^{n-1}$, otherwise $a_1 a_0^{n-1}$.
In both cases, this results in a~subset that has a~different core than $\{1'\}$, thus can be distinguished from $\{(1,s'),(1,0'),(1,1')\}$ as we showed before.\qed
\end{proof}

\bibliographystyle{splncs04}
\providecommand{\noopsort}[1]{}

\end{document}